\documentclass[letterpaper, 10 pt,journal]{IEEEtran}  \IEEEoverridecommandlockouts              
\usepackage[table,xcdraw]{xcolor}
\usepackage{amsmath}
\usepackage{amssymb}
\usepackage{amsfonts}
\usepackage{tabularx}
\usepackage{hyperref}
\usepackage{tikz}
\usetikzlibrary{matrix}
\usetikzlibrary{positioning}
\usetikzlibrary{calc}
\usepackage{graphicx}
\usepackage{caption}
\usepackage{multirow}
\usepackage{listings}
\usepackage{subcaption}
\usepackage{color}
\usepackage{algorithm}
\makeatletter
\renewcommand{\ALG@name}{Alg}
\makeatother
\usepackage{algpseudocode}
\usepackage{verbatim}
\usepackage{comment}
\usepackage{mathtools}

\usepackage{enumitem}

\newcounter{parentnumber}

\title{\LARGE \bf
Guaranteed Phase \& Topology Identification in Three Phase Distribution Grids
}
\author{Mohini Bariya$^*$, Deepjyoti Deka$^{\dagger}$, and Alexandra von Meier$^*$% <-this % 
\thanks{Corresponding authors:}
\thanks{\tt\small mohini@berkeley.edu}
\thanks{$^{*}$Dept of Electrical Engineering, University of California at Berkeley}
\thanks{$^{\dagger}$Theoretical Division, Los Alamos National Laboratory}
\thanks {This work was supported in part by NSF Award 1840083.}
}

\newtheorem{theorem}{\bf{Theorem}}

\newtheorem{lemma}{Lemma}
\newtheorem{remark}{Remark}

\newcommand\copyrighttext{%
  \footnotesize \textcopyright 2021 IEEE. Personal use is permitted, but republication/redistribution requires IEEE permission.
See \url{www.ieee.org//publications_standards//publications//rights//index.html} for more information.}
\newcommand\copyrightnotice{%
\begin{tikzpicture}[remember picture,overlay]
\node[anchor=south,yshift=10pt] at (current page.south) {\fbox{\parbox{\dimexpr\textwidth-\fboxsep-\fboxrule\relax}{\copyrighttext}}};
\end{tikzpicture}%
}

\begin{document}

\maketitle
\copyrightnotice
\thispagestyle{empty}
\pagestyle{empty}

\begin{abstract}
We present a method for joint phase identification and topology recovery in unbalanced three phase radial networks using only voltage measurements. By recovering phases and topology jointly, we utilize all three phase voltage measurements and can handle networks where some buses have a subset of three phases. Our method is theoretically justified by a novel linearized formulation of unbalanced three phase power flow and makes precisely defined and reasonable assumptions on line impedances and load statistics. We validate our method on three IEEE test networks simulated under realistic conditions in OpenDSS, comparing our performance to the state of the art. In addition to providing a new method for phase and topology recovery, our intuitively structured linearized model will provide a foundation for future work in this and other applications. 
\end{abstract}

%%%%%%%%%%%%%%%%%%%%%%%%%%%%%%%%%%%%%%%%%%%%%%%%%%%%%%%%%%%%%%%%%%%%%%%%%%%%%%%%

\section{Introduction}
\IEEEPARstart{T}{opology} identification---which determines the connectivity of network nodes from an available measurement set---is vital in the electric grid, especially in distribution networks \cite{deka2017structure}, which may be switched between multiple operating configurations. Real-time topology awareness is critical for most control and optimization approaches and is essential for detecting unintentional changes caused by faulty equipment or cyberattacks. In three phase electrical networks, phase identification---determining the phase label of each phase at each bus---is part of topology identification. Correct phase labels are important for several applications, including allocating resources to minimize phase imbalance \cite{sun2015phase}. Together, topology and phase identification determine the entire network connectivity.
The literature demonstrates several approaches to phase identification---many heuristic---on simulated or real data sets. In \cite{arya2011phase} and \cite{jayadev2016novel} power balance constraints on load measurements are used to identify customer phase connectivity; this approach suffers from multiple feasible solutions. Another class of techniques applies clustering to power or voltage measurements to categorize phases, as in \cite{blakely2019spectral} and \cite{wang2016phase}. Correlations are a popular clustering distance metric for voltage magnitudes \cite{olivier2018phase, short2012advanced, pezeshki2012consumer}. In \cite{wen2015phase}, voltage correlations are used with voltage angle differences to match phases at two buses. Though some of these methods achieve good performance, they don't provide theoretical guarantees, stymieing understanding of why and how they work, and if and when they fail. For example, clustering algorithms such as k-means \cite{macqueen1967some} can suffer from local minima that may result in incorrect solutions \cite{jain1999data}. Furthermore, a theoretical justification enables understanding how load and network characteristics, such as the radial structure and line impedances, impact the success of an approach. \cite{wang2020maximum} presents a maximum likelihood estimator for phase identification, but it requires knowledge of the grid topology and is non-convex.

Topology identification has been extensively explored in the literature. Several works use voltage correlations or covariances to recover topology, based on heuristics \cite{bariya2018data}, theoretically justified patterns in covariances \cite{bolognani2013identification, deka2020graphical}, or conditional independence tests  \cite{weng2016distributed,deka2016estimating}. Another class of approaches uses linear regression to estimates the entire network model, including topology and line impedances, but require phasor current measurements in addition to voltages \cite{moffat2019unsupervised, yuan2016inverse, yu2017patopa, cavraro2018graph, park2020learning}. All these approaches assume that the system is balanced and can be approximated by a single phase network. This may lead to errors estimating unbalanced distribution networks. \cite{liao2019unbalanced, deka2019topology} present three-phase topology estimation, possibly with incorrect phase labels, using conditional independence or mutual information tests. However the sample requirements of these methods often exceed those of direct/greedy methods such as \cite{deka2017structure}. Post-topology recovery, \cite{liao2019unbalanced} uses a distance-based metric or mixed-integer program to classify unknown phases. This still falls short of a \emph{theoretically guaranteed}, \emph{polynomial-time algorithm} for phase recovery.

\textbf{Contribution:} This work presents a \textbf{g}reedy algorithm for joint \textbf{p}hase and \textbf{t}opology identification---termed \textbf{\textit{GPT}}---surpassing the prior work in that it
\begin{itemize}
    \item is provably correct under realistic assumptions and runs in polynomial time. 
    \item is applicable to real networks, where some buses may have a subset of all three phases.
    \item requires \emph{only} voltage measurements; either phasors or magnitudes (with slight differences in performance). 
    \item utilizes voltage \emph{statistics} enabling successful phase identification even from voltage magnitudes alone and in the presence of phase shifting transformers \cite{verboomen2005phase}.
\end{itemize}
GPT is based on a linearized, multi-phase power flow model mapping nodal current injections to nodal voltage phasors, and requires radial structure and diagonal dominance of line impedance matrices to guarantee correctness. When phases are known, GPT reduces to provable, greedy multi-phase topology learning generalizing prior work in \cite{deka2017structure} for the balanced setting. When topology is known apriori, it reduces to a local approach to phase identification. We demonstrate GPT's performance and improvement over prior work in both phase and topology recovery on multiple IEEE test networks simulated in Open-DSS, an open source distribution system simulator \cite{dugan2011open}.

The paper is organized as follows. Section \ref{sec:Model} presents a linearized model for unbalanced, three-phase networks with missing phases which is the theoretical basis of our work. Section \ref{sec:Phase} theoretically justifies a proximity metric for phase identification while Section \ref{sec:Topology} justifies a distance metric for topology identification once phases have been identified. Section \ref{sec:jointPhase_Top} puts the parts together to propose GPT: an algorithm for joint phase and topology recovery. Finally, Section \ref{sec:Results} presents validation on non-linear voltage data for three IEEE distribution test networks, simulated in OpenDSS. We compare performance with algorithms in prior work, demonstrating that our approach outperforms the prior work, and is robust to non-ideal measurements. 
\begin{figure}
    \centering
    \includegraphics[width=.9\columnwidth]{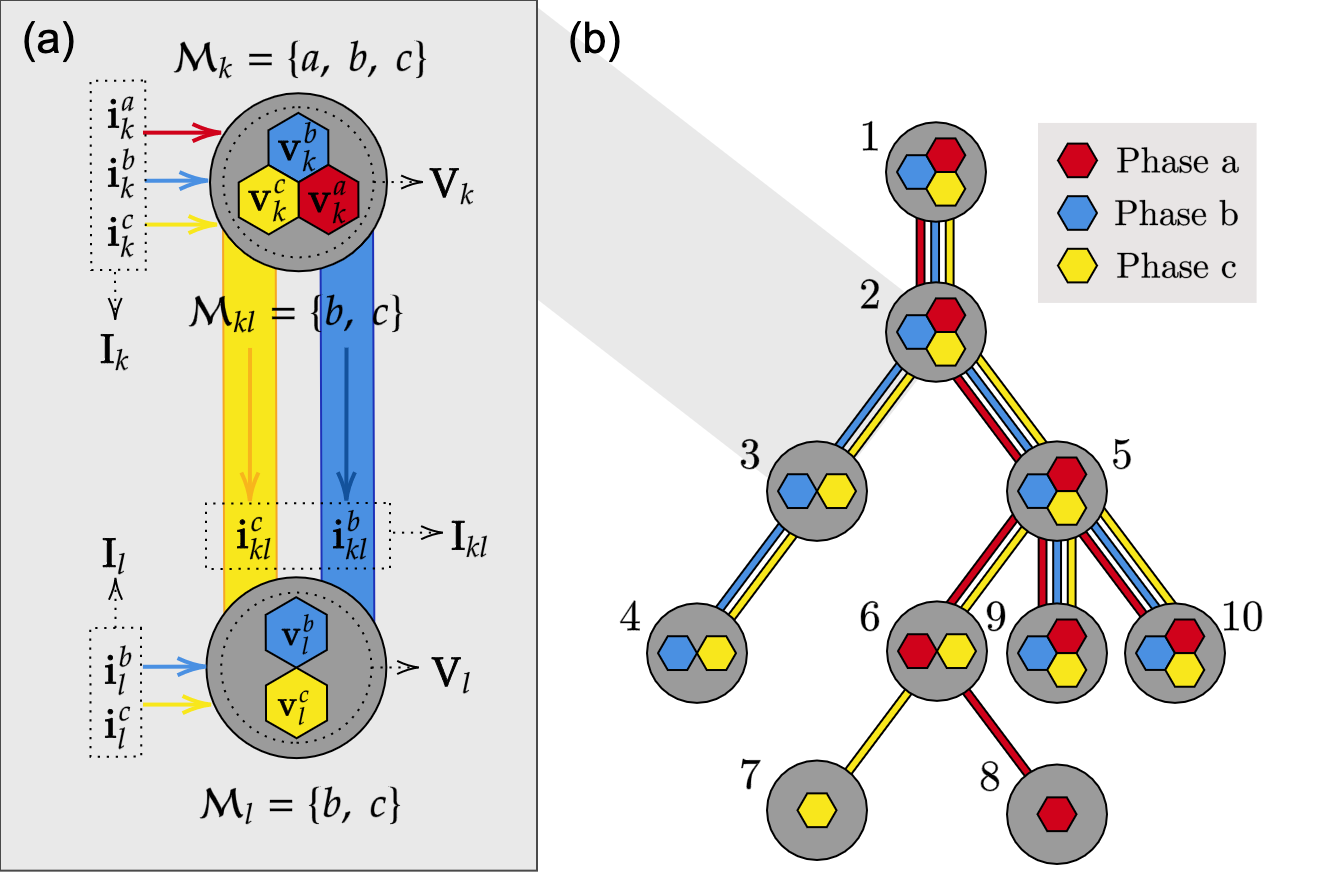}
    \caption{(a) Notation visualized for two nodes and connecting line. (b) \emph{ToyNet}: A toy network with three, two, and one phase nodes used as a running example.}
    \label{fig:example}
\end{figure}

\subsection{Notation}
Notation is summarized below and visualized in Fig. \ref{fig:example}a.

\begin{table}[h!]
\begin{tabularx}{\columnwidth}{XX|XX}
 $\mathcal{N}$ & Node set  & $\mathcal{E}$ & Edge set  \\
 $\mathbf{V}_k = \begin{bmatrix} \mathbf{v}_k^a\\ \mathbf{v}_k^b\\ \mathbf{v}_k^c\end{bmatrix}$ & Node $k$ voltages & $\mathbf{I}_k = \begin{bmatrix} \mathbf{i}_k^a\\ \mathbf{i}_k^b\\ \mathbf{i}_k^c\end{bmatrix}$  &  Node $k$ current injections\\
 $\mathcal{M}_k$ & Node $k$ phases  & $\mathcal{M}_{kl}$  & Line $kl$ phases \\
 $\mathbf{V}$ & All nodal voltages & $\mathbf{I}$ & All nodal injections\\
 $\mathbf{Y}_{kl}$ & Multiphase admittance matrix  of line $kl$ & $\mathbb{Y}$ & Network admittance matrix\\
 $\mathbb{Y}_{k,l}$ & $\mathbb{Y}$ block for nodes $k$ \& $l$ &  $\mathbb{Y}_{k,l}^{\phi,\psi}$ & $\mathbb{Y}$ element for phase $\phi$ at node $k$ \& phase $\psi$ at node $l$\\
 $A^T$ & Transpose & $A^H$ & Conjugate transpose
\end{tabularx}
\end{table}

\section{Unbalanced Three Phase Model}\label{sec:Model}
Before introducing our unbalanced three phase network model, we review the \emph{single phase balanced} power flow model. In the single phase case, each line $ij$ has an associated scalar admittance $\mathbf{y}_{ij}$ and impedance $\mathbf{z}_{ij} = \mathbf{y}^{-1}_{ij}$. The single phase voltage phasors and current injections are related by the $(|\mathcal{N}| \times |\mathcal{N}|)$ system \emph{admittance matrix} $\mathbb{Y}$ as $\mathbf{I} = \mathbb{Y}\mathbf{V}$ with the form: 
\begin{gather*}
    \mathbb{Y}_{i,j} = -\mathbf{y}_{ij},\ \mathbb{Y}_{i,i} = \sum_{ij \in \mathcal{E}} \mathbf{y}_{ij} \Rightarrow \mathbb{Y} = A\mathbf{D}A^T
\end{gather*}
$\mathbb{Y}$ can be factored into  $|\mathcal{N}|\times |\mathcal{E}|$ \emph{incidence} matrix $A$ and diagonal line impedance matrix $\mathbf{D}$ \cite{li1998laplacian}. Without loss of generality, we choose all edges to be directed away from the network ``root'', generally the point of common coupling (PCC) or substation. If edge $ij$ is oriented from $i$ to $j$, the corresponding elements of $A$ are
\begin{gather*}
    A_{i, ij} = 1,\ A_{j, ij} = -1, \ A_{r, ij} = 0 \text{~if~} r\neq i\neq j
\end{gather*}
where $A_{i, ij}$ is the $i^{th}$ element of the column corresponding to edge $ij$. 
By definition, as $\mathbb{Y}1 = 0$, $\mathbb{Y}$ is not invertible. An invertible reduced admittance matrix, $\underline{\mathbb{Y}}$, is constructed by choosing a reference node $r$ and removing the corresponding row and column of $\mathbb{Y}$. Since the system is lossless, its inverse relates voltages and currents as follows: 
\begin{align}
    \underline{\mathbb{Z}} = \underline{\mathbb{Y}}^{-1}, ~~ \underline{\mathbf{V}} = \underline{\mathbb{Z}}\underline{\mathbf{I}}
\end{align}
$\underline{\mathbf{V}}$  contains voltages differences to the reference voltage while $\underline{\mathbf{I}}$ contains current injections at non-reference nodes. Let $\mathcal{E}_i$ and $\mathcal{E}_j$ denote the edge sets on the unique path in the radial system to $r$ from nodes $i$ and $j$ respectively. The value of $\underline{\mathbb{Z}}_{i,j}$ is given by : 
\begin{align*}
    \underline{\mathbb{Z}}_{i,j} = \sum_{kl \in (\mathcal{E}_i \cap \mathcal{E}_j)} \mathbf{z}_{kl}. 
\end{align*}
Thus, the elements of $\underline{\mathbb{Z}}$ correspond to the impedances of common paths between node pairs and the reference \cite{moffat2019unsupervised,deka2017structure}.

The \emph{unbalanced three phase model} follows from the single phase one. To clarify definitions, we use \emph{ToyNet} (Fig. \ref{fig:example}b), a simple, unbalanced, three phase radial network, as a running example. 
We begin with the model for a multiphase line $ij$, with phases $\mathcal{M}_{ij} \subseteq \{a,b,c\}$. The voltage across $ij$ is related to the current along each phase of the line by line impedance matrix $\mathbf{Y}_{ij}$:
\begin{align}\label{eq:line_impedance}
    \mathbf{I}_{ij} = \mathbf{Y}_{ij}(\mathbf{V}_i - \mathbf{V}_j) 
\end{align}
$\mathbf{Y}_{ij}$ is the inverse of the {$(|\mathcal{M}_{ij}| \times |\mathcal{M}_{ij}|)$} line impedance matrix, $\mathbf{Z}_{ij}$. Eq. (\ref{eq:line_impedance}) for line $56$ in ToyNet is: {$\begin{bmatrix}\textbf{i}_{56}^a \\ \textbf{i}_{56}^c \end{bmatrix} = \mathbf{Y}_{56}\begin{bmatrix}\textbf{v}_5^a-\textbf{v}_6^a \\ \textbf{v}_5^c-\textbf{v}_6^c \end{bmatrix}$.} The node $i$ current injection, denoted $\mathbf{I}_i$, is a vector of injections on each phase of $i$, and is given by the sum of line flows: $\mathbf{I}_i = \sum_{ij \in \mathcal{E}} \mathbf{I}_{ij}$.
Building up from the current-voltage relations across individual lines in (\ref{eq:line_impedance}), the multi-phase voltages and currents injections across the network are related by $\mathbf{I} = \hat{\mathbb{Y}}\mathbf{V}$. Note this model can describe a network with a subset of phases at some nodes. $\hat{\mathbb{Y}}$ is the multi-phase system admittance matrix with dimensions $(\sum_{i \in \mathcal{N}} |\mathcal{M}_i|) \times (\sum_{i \in \mathcal{N}} |\mathcal{M}_i|)$. The $i,j$ \emph{block} of $\hat{\mathbb{Y}}$ is
\begin{align}\label{eq:3phaseYdef}
    \hat{\mathbb{Y}}_{i,j} = -\mathbf{Y}_{ij},\ \hat{\mathbb{Y}}_{i,i} = \sum_{ij \in \mathcal{E}} \mathbf{Y}_{ij} 
\end{align}
Block $\hat{\mathbb{Y}}_{i,j}$ is $(|\mathcal{M}_i|\times |\mathcal{M}_j|)$, so $\mathbf{Y}_{ij}$ must be appropriately zero-padded or reduced if $i$ and $j$ don't have all the same phases. For ToyNet, $\hat{\mathbb{Y}}$ has the structure visualized in Fig. \ref{fig:modelViz}.
\begin{remark}
$\hat{\mathbb{Y}}$ can be factored into an incidence matrix $\hat{A}$, which captures the endpoints of each edge, and $\hat{\mathbf{D}}$, a block diagonal matrix of line admittances: $\hat{\mathbb{Y}} = \hat{A}\hat{\mathbf{D}}\hat{A}^T.$
$\hat{\mathbf{D}}$ has dimensions $(\sum_{ij \in \mathcal{E}}|\mathcal{M}_{ij}|) \times (\sum_{ij \in \mathcal{E}}|\mathcal{M}_{ij}|)$, with line admittance matrices $\mathbf{Y}_{ij}$ along the diagonal. $\hat{A}$ is $(\sum_{i \in \mathcal{N}}|\mathcal{M}_i|) \times (\sum_{ij \in \mathcal{E}}|\mathcal{M}_{ij}|)$ dimensional. Its rows correspond to phases at each bus, and columns to phases of each edge. With edges directed toward the root, assume edge $ij \in \mathcal{E}$ is oriented from $i$ to $j$. Then for every $ij \in \mathcal{E}$, with $\phi\in\mathcal{M}_{ij}$: $\hat{A}_{i, ij}^{\phi\phi} = 1,\ \hat{A}_{j, ij}^{\phi\phi} = -1$.
All other elements of $\hat{A}$ are zero. $\hat{A}$ for ToyNet is visualized in Fig. \ref{fig:modelViz}.
\end{remark}
\begin{figure}
    \centering
    \includegraphics[width=\columnwidth]{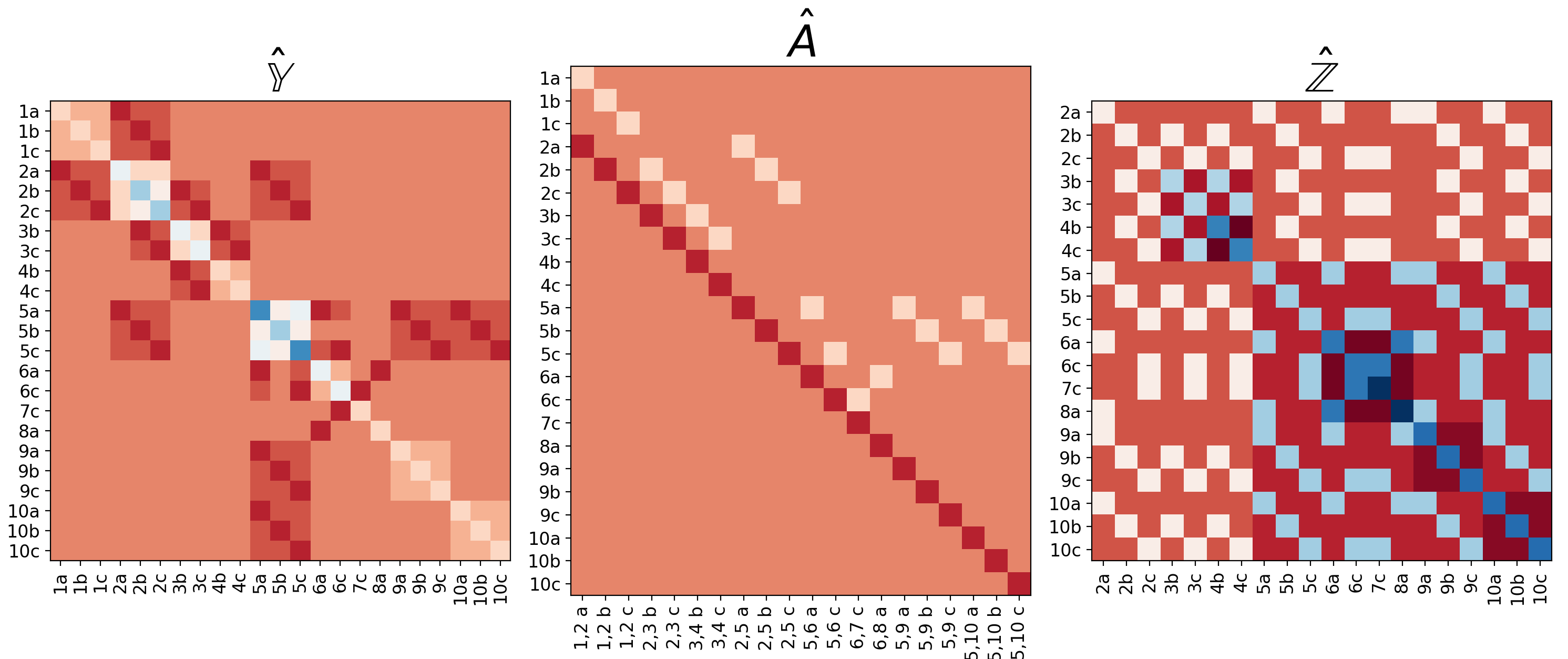}
    \caption{Visualizing the structure of three phase admittance, adjacency and impedance matrices for ToyNet.}
    \label{fig:modelViz}
\end{figure}

\subsection{Inverting the model}
$\hat{\mathbb{Y}}$ maps voltages to current injections, but we use the inverse mapping for phase and topology identification. By definition (\ref{eq:3phaseYdef}), $\hat{\mathbb{Y}}$ is singular. Again, a reduction denoted $\underline{\hat{\mathbb{Y}}}$ is invertable. To obtain $\underline{\hat{\mathbb{Y}}}$, we remove the \emph{three} rows and columns of $\hat{\mathbb{Y}}$ corresponding to the three phases at reference node $r$. $\underline{\hat{\mathbb{Y}}}$ can be factored as: $\underline{\hat{\mathbb{Y}}} = \underline{A}\hat{\mathbf{D}}\underline{A}^T$ where $\underline{A}$ is obtained from $\hat{A}$ by deleting the three rows corresponding to $r$.
To derive $\underline{\hat{\mathbb{Y}}}^{-1}$, we begin with the right pseudoinverse of $\underline{A}^T$, which has the following properties. 
\begin{lemma}\label{lemma:B}
Let $\underline{B}$ be the right pseudoinverse of $\underline{A}^T$, with rows corresponding to nodes and columns to edges. For $i \in \mathcal{N}$ with phases $\mathcal{M}_i$, let $\mathcal{E}_i$ be the edge set of the unique path to $r$. Then,
\begin{align*}
    \underline{B}_{i, kl}^{\phi\psi} =  
\begin{cases} -1 ~~\forall \phi =\psi \in \mathcal{M}_i,\ \forall kl \in \mathcal{E}_i\\ 0 \textnormal{~~~~ otherwise } \end{cases} 
\end{align*} 
\end{lemma}
\begin{proof}
$\underline{B}_{i}^{\phi}$ is the row of $\underline{B}$ corresponding to phase $\phi$ at node $i$, while $\underline{B}_{,ij}^{,\phi}$ is the column corresponding to phase $\phi$ of edge $ij$. If $\phi \neq \psi$, $\underline{B}_{,ij}^{,\phi T}\underline{A}_{,kl}^{,\psi} =0$. A column of $\underline{A}$ has only two nonzero elements, so for $\phi=\psi$, we have
\begin{align*}
    \underline{B}_{,ij}^{,\phi T}\underline{A}_{,kl}^{,\phi} = -1\delta(ij\in\mathcal{E}_k) + 1\delta(ij\in\mathcal{E}_l)
\end{align*}
For any edge $kl \neq ij$, we will have either $(ij\in \mathcal{E}_k), (ij \in \mathcal{E}_l)$ \emph{or} $(ij\not\in \mathcal{E}_k), (ij \not\in \mathcal{E}_l)$. Thus $\underline{B}_{,ij}^{,\phi T}\underline{A}_{,kl}^{,\phi} =0$ for $kl \neq ij$. If $kl=ij$, we have $(ij\not\in \mathcal{E}_k), (ij \in \mathcal{E}_l)$ and $\underline{B}_{,ij}^{,\phi T}\underline{A}_{,ij}^{,\phi}=1$. Thus,  
$\underline{B}_{,ij}^{,\phi T}\underline{A}_{,kl}^{,\psi} = 1 \textnormal{ iff }  ij=kl, \phi = \psi \Rightarrow \underline{B}^T\underline{A} = I.$

Now consider, $\underline{B}_{i}^{\phi}\underline{A}_{j}^{\psi T}$, the inner product of rows. If $\phi\neq \psi$, this is $0$. Consider when $\phi=\psi$ and $i\neq j$. If $j$ does not lie along the path from $i$ to the reference, then $jk \not\in \mathcal{E}_i$, i.e. there is no edge connected to $j$ in $\mathcal{E}_i$, and $\underline{B}_{i}^{\phi}\underline{A}_{j}^{\phi T} = 0$. In contrast, if $j$ lies along the path from $i$ to the reference, there \emph{must} be two edges $kj, jl \in \mathcal{E}_i$ oriented to and away from $j$ respectively, as the path passes through $j$. Then, $\underline{B}_{i}^{\phi}\underline{A}_{j}^{\phi T} = (-1\times -1) + (1\times -1) = 0$. If $i = j$, only edge $il =jl \in \mathcal{E}_i$, and  $\underline{B}_{i}^{\phi}\underline{A}_{i}^{\phi T} = 1$. Therefore, $\underline{B}_{i}^{\phi}\underline{A}_{j}^{\psi T} = 1 \textnormal{ iff } i=j, \phi=\psi$  
\end{proof}
Therefore, ${\underline{\hat{\mathbb{Y}}}}^{-1} = \hat{\underline{\mathbb{Z}}}$ can be written as follows. 
\begin{theorem}\label{thm:Y_inv}
The inverse of $\underline{\hat{\mathbb{Y}}}$ is given by:
\begin{align}\label{eq:Zdef}
    \underline{\hat{\mathbb{Z}}} = {\underline{\hat{\mathbb{Y}}}}^{-1} = \underline{B}\hat{\mathbf{D}}^{-1}\underline{B}^{T}
\end{align}
where $\hat{\mathbf{D}}^{-1}$ is block diagonal matrix of line impedance matrices: $\hat{\mathbf{D}}^{-1}_{ij,ij} = \mathbf{Z}_{ij}$. Further, the element of $\underline{\hat{\mathbb{Z}}}$ corresponding to phase $\phi$ at node $i$ and phase $\psi$ at node $j$ is given by:
\begin{align}\label{eq:Zelements}
    \underline{\hat{\mathbb{Z}}}_{ij}^{\phi\psi} = \sum_{kl \in (\mathcal{E}_i \cap \mathcal{E}_j)}\mathbf{Z}_{kl}^{\phi\psi}
\end{align}
\end{theorem}
\begin{proof}
Using the structure of $\underline{B}$ from Lemma \ref{lemma:B}, we have 
$
    \underline{\hat{\mathbb{Z}}}\underline{\hat{\mathbb{Y}}} = (\underline{B}\hat{\mathbf{D}}^{-1}\underline{B}^{T})(\underline{A}\hat{\mathbf{D}}\underline{A}^{T}) = I
$. Thus $ \underline{\hat{\mathbb{Z}}} = {\underline{\hat{\mathbb{Y}}}}^{-1}$.

Consider the block of $\underline{\hat{\mathbb{Z}}}$ corresponding to nodes $i$ and $j$. Based on Lemma \ref{lemma:B}:
\begin{align}
    \underline{\hat{\mathbb{Z}}}_{ij}^{\phi\psi} &= \sum_{kl \in \mathcal{E}} \underline{B}_{i,kl}^{\phi\phi}\mathbf{Z}_{kl}^{\phi\psi}\underline{B}_{j,kl}^{\psi\psi T} &= \sum_{kl \in (\mathcal{E}_i \cap \mathcal{E}_j)}\mathbf{Z}_{kl}^{\phi\psi}\nonumber
\end{align}
\end{proof}
Intuitively, (\ref{eq:Zelements}) says that a change in current injection on phase $\psi$ at node $j$ will affect the voltage at phase $\phi$ at node $i$, proportional to the $(\phi\psi)$ impedance of the shared path $(\mathcal{E}_i \cap \mathcal{E}_j)$ from $i,j$ to $r$. In our definition, $\underline{\hat{\mathbb{Y}}}$ and $\underline{\hat{\mathbb{Z}}}$ are ordered with the phases of each node or edge grouped together. If all phases exist at all nodes, this is equivalent to a permutation of the three-phase model in \cite{deka2019topology}, where entries for one phase across all nodes and edges are grouped together.

Using this theorem, in the unbalanced three phase model, voltages are related to currents as: 
\begin{align}\label{eq:3phaseOhms}
    \underline{\mathbf{V}} = \underline{\hat{\mathbb{Z}}}\underline{\mathbf{I}}
\end{align}
Here, $\underline{\mathbf{V}}$ contains nodal voltage \emph{differences} with the voltage for the \emph{matching phase} at the reference: $\underline{\mathbf{V}}_i^{\phi} = \mathbf{v}_i^{\phi}-\mathbf{v}_r^{\phi}$. $\underline{\mathbf{I}}$ contains all current injections except at the reference. Note that there is no assumption for all nodes having all three phases. In the next section, we will use the model of (\ref{eq:3phaseOhms}) to determine patterns in voltage statistics to enable phase and topology recovery.

\subsection{Current and Voltage Statistics}
We treat voltages as random variables driven by current via the model of (\ref{eq:3phaseOhms}). In our theoretical analysis, we assume
\begin{enumerate}[leftmargin=*]
\item current injections are uncorrelated across nodes and phases (including at a single node). \begin{align}\label{eq:assumption1}
        cov(\mathbf{i}_i^{\phi}, \mathbf{i}_j^{\psi}) \neq 0 \textnormal{ iff } (i=j) \cap (\phi=\psi)
    \end{align}
    \item current injections have equal variance at all nodes. 
    \begin{align}\label{eq:assumption2}
        \forall i, \phi:\ var(\mathbf{i}_i^{\phi}) = s^2
    \end{align}
\end{enumerate}
As they are predominantly determined by loads---which are uncorrelated over time intervals on the order of seconds---current injections can be modeled as uncorrelated across nodes and phases when using high resolution measurements such as from PMUs. We assume PMUs report at $120 Hz$, but our methods apply if resolution is sufficient for measurements to be de-trended to remove inter-nodal correlations. Assumption (\ref{eq:assumption2}) is stronger but permissible in reasonably balanced networks. In Section \ref{sec:Results}, we evaluate how deviations from these assumptions impact recovery performance. 

\section{Voltage Covariance for Phase Matching}\label{sec:Phase}
Voltages covariances are informative for phase identification. Under Assumptions (\ref{eq:assumption1},\ref{eq:assumption2}), the covariance of the voltage of phase $\phi$ at node $i$ and phase $\psi$ at node $j$, in the three-phase model (\ref{eq:3phaseOhms}), is given by: 
\begin{align}\label{eq:covarv}
    cov(\mathbf{v}_i^{\phi}, \mathbf{v}_j^{\psi}) = cov(\underline{\hat{\mathbb{Z}}}_i^{\phi} \underline{\mathbf{I}},\underline{\hat{\mathbb{Z}}}_j^{\psi}\underline{\mathbf{I}}) = s^2Re((\underline{\hat{\mathbb{Z}}}_i^{\phi})^H\underline{\hat{\mathbb{Z}}}_j^{\psi})
\end{align}
We are interested in the sum of covariances for a particular \emph{phase ordering} between nodes $i$ and $j$. Consider the case where the phases at $i$ are a subset of those at $j$ ($\mathcal{M}_i \subseteq \mathcal{M}_j$). Let $\mathcal{O}$ denote the ordering/permutation of phases at $j$, where $\mathcal{O}(\phi)$ denotes the specific phase at $j$ matched to the phase $\phi$ at $i$. Then, the covariance sum for matching $\mathcal{O}$, denoted by $c_{ij}^{\mathcal{O}}$, is:
\begin{align}\label{eq:distMetric}
    c_{ij}^{\mathcal{O}} = \sum_{\phi \in \mathcal{P}_i} cov(\mathbf{v}_i^{\phi}, \mathbf{v}_j^{\mathcal{O}(\phi)})
\end{align}
Let $\underline{\hat{\mathbb{Z}}}_i$ denote the rows of $\underline{\hat{\mathbb{Z}}}$ corresponding to \emph{all} phases at node $i$, and $\underline{\hat{\mathbb{Z}}}_j^{\mathcal{O}}$ denote the rows corresponding to the phases at $j$ ordered according to $\mathcal{O}$. Then $c_{ij}^{\mathcal{O}}$ is:
\begin{align}
    c_{ij}^{\mathcal{O}} &= s^2Re(vec(\underline{\hat{\mathbb{Z}}}_i)^{H}vec(\underline{\hat{\mathbb{Z}}}_j^{\mathcal{O}})) = s^2Re(\sum_{k \in \mathcal{N}}vec(\underline{\hat{\mathbb{Z}}}_{ik}^H)vec(\underline{\hat{\mathbb{Z}}}_{jk}^{\mathcal{O}}))\nonumber\\
    &= s^2Re\Bigg[\sum_{k\in\mathcal{N}}\smashoperator[r]{\sum_{\substack{mn\in(\mathcal{E}_i \cap \mathcal{E}_k)\\pq\in(\mathcal{E}_j\cap\mathcal{E}_k)}}}vec(\mathbf{Z}_{mn}^{\mathcal{M}_i})^Hvec(\mathbf{Z}_{pq}^{\mathcal{O}})\Bigg]\label{eq:C_ijexpand}
\end{align}
The last equality follows from (\ref{eq:Zelements}). The contribution of a node $k$ to $c_{ij}^{\mathcal{O}}$ is the dot product of the common path lengths between $i,k$ and $j,k$. The following result shows how $c_{ij}^{\mathcal{O}}$ enables phase matching.

\begin{theorem}\label{thm:maxCov}
Consider $c_{ij}^{\mathcal{O}}$ given in (\ref{eq:distMetric}) for $\mathcal{M}_i \subseteq \mathcal{M}_j$. If condition (\ref{eq:lineCond}) holds for each pair of line impedance matrices, then $c_{ij}^{\mathcal{O}}$ is maximized when $\mathcal{O}$ corresponds to the correct phase matching between $i$ and $j$.
\begin{align}
    \forall \mathcal{M} \in \{\mathcal{M}_1,...,\mathcal{M}_n\},\ \forall st, kl \in \mathcal{E}:\nonumber\\ \mathcal{M} = \arg\max_{\mathcal{O}} Re\bigg[vec(\mathbf{Z}_{st}^{\mathcal{M}})^Hvec(\mathbf{Z}_{kl}^{\mathcal{O}(\mathcal{M})})\bigg]\label{eq:lineCond}
\end{align}
where $\mathcal{M}$ ranges over every nodal phase set $\mathcal{M}_i$.
\end{theorem}
\begin{proof} If (\ref{eq:lineCond}) holds, \emph{every} term in the summation in (\ref{eq:C_ijexpand}) is maximized by $\mathcal{O} = \mathcal{M}_i$. Therefore, $\mathcal{M}_i$ maximizes the sum, and $\mathcal{M}_i = \arg\max_{\mathcal{O}} c_{ij}^{\mathcal{O}}$
\end{proof}
(\ref{eq:lineCond}) is a condition on every pair of edges, $st$ and $kl$, in the network. It states that for every row subset of line impedance matrix $\mathbf{Z}_{st}$ (corresponding to each nodal phase set $\mathcal{M}_1,...,\mathcal{M}_n$), the matching rows of $\mathbf{Z}_{kl}$ produce the largest vectorized dot product. This is reasonable as real line impedance matrices are diagonally dominant. Condition (\ref{eq:lineCond}) depends on the particular network considered. In a network where all nodes have all three phases ($\mathcal{M}=\mathcal{M}_i = \{a,b,c\}$) the condition on the vectorized dot product involves the full three phase line impedance matrices. if some nodes have a subset of phases, it will involve sub-matrices of impedance matrices.
Note that, in general, $cov(\mathbf{v}_i^{\phi}, \mathbf{v}_j^{\psi}) > 0$ even if $\phi \neq \psi$. If node $i$ has phases $a,b$ and node $j$ has phases $a,c$ ($\mathcal{M}_i \not\subseteq \mathcal{M}_j$), minimizing $c_{ij}$ will incorrectly match $b, c$. GPT avoids such scenarios by ordering nodes; discussed later.
In summary, Theorem \ref{thm:maxCov} allows us to use $c_{ij}^{\mathcal{O}}$ as a proximity metric for phase matching. 

\section{Voltage Difference Variances for Topology}\label{sec:Topology}
We use voltage difference variances for topology recovery. Define $d_{ij}$ to be the sum of the variance of the voltage differences between correctly matched phases of nodes $i,j$. Assuming $\mathcal{M}_i \subseteq \mathcal{M}_j$, $d_{ij}$ is: 
\begin{align}\label{eq:sovDist}
     d_{ij} = \smashoperator[lr]{\sum_{\phi \in \mathcal{M}_i}}var(\mathbf{v}_i^{\phi}-\mathbf{v}_j^{\phi}) = \smashoperator[lr]{\sum_{\phi \in \mathcal{M}_i}}\mathbb{E}[(\mathbf{v}_i^{\phi}-\mathbf{v}_j^{\phi})-\mathbb{E}(\mathbf{v}_i^{\phi}-\mathbf{v}_j^{\phi})]^2 
\end{align}
Lemma \ref{lem:minVarDiff} establishes trends in $d_{ij}$ along one phase.
\begin{lemma}\label{lem:minVarDiff}
Given the voltage on phase $\phi$ at node $i$:
\begin{align}\label{eq:minDist}
    \arg\min_{j} d_{ij}^{\phi} \triangleq 
    \arg\min_{j} var(\mathbf{v}_i^{\phi} - \mathbf{v}_j^{\phi}) \in \textnormal{Parent/Child of } i
\end{align}
\end{lemma}
\begin{proof}
Expanding the difference, we obtain:
\begin{align}\label{eq:varsum}
    d_{ij}^{\phi} = var(\mathbf{v}_i^{\phi}-\mathbf{v}_j^{\phi}) = \sum_{n\in\mathcal{N}}\sum_{\psi\in\mathcal{M}_n}s_n^2 |\underline{\hat{\mathbb{Z}}}_{in}^{\phi\psi} - \underline{\hat{\mathbb{Z}}}_{jn}^{\phi\psi}|^2
\end{align}
where $s_n$ is the injection variance at node $n$. If the paths from nodes $k$ and $l$ to $r$ merge at node $n$, $\underline{\hat{\mathbb{Z}}}_{k,l}^{\phi,\psi} = \mathbf{e}_n^{\phi\psi}$, the impedance of the path from $n$ to $r$ along phase coupling $\phi, \psi$: 
\begin{align}\label{eq:path}
\mathbf{e}_n^{\phi\psi} = \sum_{ij \in \mathcal{E}_n} \mathbf{Z}_{ij}^{\phi\psi}
\end{align}
To determine the minimizer of (\ref{eq:minDist}), consider two cases visualized in Fig. \ref{fig:AllRegions}. In case A, $j$ is the common ancestor of node $i,k$ on the path to the root. In case B, $j$ is an ancestor of $i$, while $k$ is an ancestor of $j$. In both cases, we show that $d_{ij}^{\phi},d_{jk}^{\phi} < d_{ik}^{\phi}$. Put together, for a given $i$, the minimizer $j$ of $d_{ij}^{\phi}$ is either the parent or child of $i$.

\textbf{Case A.} We split the sum in (\ref{eq:varsum}) into the regions $\mathcal{N}_i$ in Fig. \ref{fig:AllRegions}a. Using (\ref{eq:path}) in (\ref{eq:varsum}) for each region, we have
\begin{align*}
    d^{\phi}_{ik}-d_{ij}^{\phi} = 
    \sum_{n \in \mathcal{N}_1,\psi \in \mathcal{M}_n} 0 
    + \sum_{n \in \mathcal{N}_2,\psi \in \mathcal{M}_n}0 + \sum_{n \in \mathcal{N}_4,\psi \in \mathcal{M}_n}0\\ 
    + \sum_{n \in \mathcal{N}_3,\psi \in \mathcal{M}_n}s_n^2\bigg(|\mathbf{e}_j^{\phi\psi}-\mathbf{e}_n^{\phi\psi}|^2-|\mathbf{e}_j^{\phi\psi}-\mathbf{e}_j^{\phi\psi}|^2\bigg)\\
    + \sum_{n \in \mathcal{N}_5,\psi \in \mathcal{M}_n}s_n^2\bigg(|\mathbf{e}_j^{\phi\psi}-\mathbf{e}_k^{\phi\psi}|^2-|\mathbf{e}_j^{\phi\psi}-\mathbf{e}_j^{\phi\psi}|^2\bigg) > 0
\end{align*}
A similar argument shows $ d_{ik}^{\phi}-d_{kj}^{\phi}> 0$.\\
\textbf{Case B.} Now we split (\ref{eq:varsum}) over the regions in Fig. \ref{fig:AllRegions}b. Using (\ref{eq:path}), we have
\begin{align*}
    d_{ik}^{\phi}-d_{ij}^{\phi} = 
    \sum_{n \in \mathcal{N}_1,\psi \in \mathcal{M}_n}0
    + \sum_{n \in \mathcal{N}_2,\psi \in \mathcal{M}_n}0\\ 
    + \sum_{n \in \mathcal{N}_3,\psi \in \mathcal{M}_n}s_n^2\bigg(|\mathbf{e}_n^{\phi\psi}-\mathbf{e}_k^{\phi\psi}|^2-|\mathbf{e}_n^{\phi\psi}-\mathbf{e}_n^{\phi\psi}|^2\bigg)\\
    + \sum_{n \in \mathcal{N}_4,\psi \in \mathcal{M}_n}s_n^2\bigg(|\mathbf{e}_j^{\phi\psi}-\mathbf{e}_k^{\phi\psi}|^2 - |\mathbf{e}_j^{\phi\psi}-\mathbf{e}_j^{\phi\psi}|^2\bigg)\\ 
    + \sum_{n \in \mathcal{N}_5,\psi \in \mathcal{M}_n}s_n^2\bigg(|\mathbf{e}_n^{\phi\psi}-\mathbf{e}_k^{\phi\psi}|^2-|\mathbf{e}_n^{\phi\psi}-\mathbf{e}_j^{\phi\psi}|^2\bigg)\\
    + \sum_{n \in \mathcal{N}_6,\psi \in \mathcal{M}_n}s_n^2\bigg(|\mathbf{e}_i^{\phi\psi}-\mathbf{e}_k^{\phi\psi}|^2-|\mathbf{e}_i^{\phi\psi}-\mathbf{e}_j^{\phi\psi}|^2\bigg) > 0
\end{align*}
A similar analysis shows $d_{ik}^{\phi}-d_{kj}^{\phi}> 0$. Thus the minimum is given by the parent/child of $i$. 
\end{proof}
\begin{figure}
    \centering
    \includegraphics[width=0.9\columnwidth]{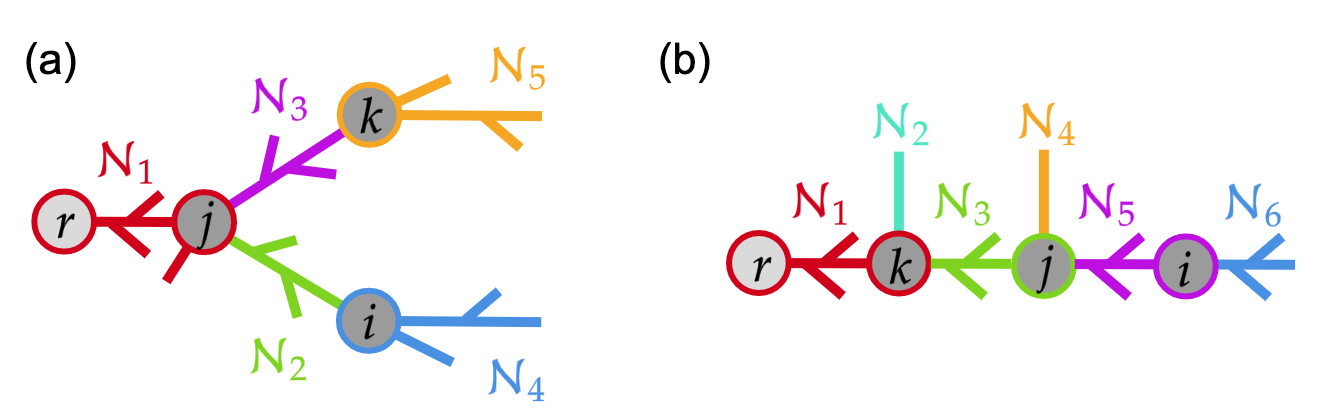}
    \caption{Regions of the radial network when (a) Case A: $k$ lies off the path between $i$ and the reference. (b) Case B: $i$ lies along the path from $j$ to the reference.}
    \label{fig:AllRegions}
\end{figure}Applying Lemma \ref{lem:minVarDiff} to all matched phases between two nodes gives the following result:
\begin{theorem}\label{thm:topology}
Given node $i$, the node $j$ which minimizes $d_{ij}$ in (\ref{eq:sovDist}) is either a parent or child of $i$.
\end{theorem}
If phases at each node are known, Thm. \ref{thm:topology} enables correct topology recovery with a greedy algorithm based on distance $d_{ij}$. Note that Lemma \ref{lem:minVarDiff} and Thm. \ref{thm:topology} hold for all uncorrelated injections even with unequal variances. Thus, Assumption \ref{eq:assumption2} can be relaxed for topology learning. We now have the tools for joint phase and topology recovery, detailed in the next section.

\section{Joint Phase \& Topology Identification}\label{sec:jointPhase_Top}
We propose GPT (Alg. \ref{alg:greedy}): a \emph{greedy} algorithm for joint phase and topology identification based on the nodal voltage properties of Sections \ref{sec:Phase}-\ref{sec:Topology}. GPT computes $c_{ij}^{\mathcal{O}}$'s (\ref{eq:covarv}) exhaustively (for all matching options), selecting maxima for phase matching (Theorem \ref{thm:maxCov}). Based on phase matchings, it computes $d_{ij}$'s (\ref{eq:sovDist}) exhaustively (for all node pairs), selecting minima for topology recovery (Theorem \ref{thm:topology}). GPT greedily builds a tree with node set $\mathcal{T}$, starting from node $i_0$ and iterating till all nodes have been added. In each iteration, a new node is added to the tree by choosing node $i \not\in \mathcal{T}$, which has the minimum value of $d_{ij}$ for all $j \in \mathcal{T}$, using the \emph{getNext} algorithm. When node $i$ is added to the tree, it's phases are determined based on their matching to the phases of node $j$. 
%%%%%%%%%%%%%%%%%%%%%%%%%%%%%%%%%%%%%%%%%%%%%
% ALGORITHMs
\begin{algorithm}
\caption{\emph{GPT} Greedy topology and phase recovery}\label{alg:greedy}
\textbf{Input:} Multi-phase voltage time series for all nodes in $V$.\\
\textbf{Output:} $E$: network edge set, $P$: phase ordering of each node. \\
\begin{algorithmic}[1]
\State $N3,N2,N1 \gets$ set of three, two, one phase nodes in $V$.
\ForAll{$i \neq j \in N$}
\State $\forall \mathcal{O}$ compute $c_{ij}^{\mathcal{O}}$
\State $M[i, j] \leftarrow \arg\max c_{ij}^{\mathcal{O}}$ in (\ref{eq:distMetric}) \Comment{Phase matching}
\State $D[i, j] \leftarrow$ $d_{ij}$ in (\ref{eq:sovDist}) with matched $M[i, j]$ 
\EndFor
\State $T  \leftarrow \{i_0\}, N3 \leftarrow N3 \setminus i_0$,
 $P[i_0] \leftarrow [a, b, c]$ \Comment{Add first three phase node to tree $T$ and set phases}
\While{$N3 \not= \emptyset$}\Comment{Connect 3 phase nodes}
 \State $i$, $j$ $\leftarrow$ getNext($D$, $T$, $N3$)
 \State $T \leftarrow T \cup i, N3 \leftarrow N3 \setminus i$, $E \leftarrow E \cup e_{ij}$, {$P[i] \leftarrow M[i, j]$}
 \EndWhile
 \While{$N2 \not= \emptyset$}\Comment{Connect 2 phase nodes}
 \State $i$, $j$ $\leftarrow$ getNext($D$, $T$, $N2$)
 \State $T \leftarrow T \cup i$, $N2 \leftarrow N2 \setminus i$, $E \leftarrow E \cup e_{ij}$, {$P[i] \leftarrow M[i, j]$}
 \EndWhile
 \While{$N1 \not= \emptyset$}\Comment{Connect 1 phase nodes}
 \State $i$, $j$ $\leftarrow$ getNext($D$, $T$, $N1$)
 \State $T \leftarrow T \cup i$, $N1 \leftarrow N1 \setminus i$, $E \leftarrow E \cup e_{ij}$, 
 {$P[i] \leftarrow M[i, j]$}
 \EndWhile
\end{algorithmic}
\end{algorithm}
\begin{algorithm}\label{alg:getNext}
\caption{\emph{getNext}}
\textbf{Input:} Pairwise distances in $D$, Added nodes in $T$, Nodes to add in $N$\\
\textbf{Output:} $i \in N$ to be added, $j \in T$ connected to $i$. \\
\begin{algorithmic}[1]
\State $d_{ij} \leftarrow \infty$, $i \leftarrow$ None, $j \leftarrow$ None\;
\ForAll{$b \in T, a \in N$}
\If{$D[a, b] < d_{ij}$}
\State $d_{ij} \leftarrow D[a, b], i \leftarrow a, j \leftarrow b$
\EndIf
\EndFor
\end{algorithmic}
\end{algorithm}
%%%%%%%%%%%%%%%%%%%%%%%%%%%%%%%%%%%%%%%%%%%%%
GPT adds $3$ phase, then $2$ phase, then $1$ phase nodes to the tree. The initial node must be three phase, making the reference an intuitive choice. By adding nodes in this order, GPT implicitly enforces the crucial fact that number of phases never increases moving from the substation to the network ends (a single phase node is never the parent of a three phase node) and avoids issues that can arise when applying a naive greedy algorithm to a network with a variable number of phases at each node. For example, suppose we are recovering the topology of \emph{ToyNet}. All nodes have been added to $\mathcal{T}$ except $6$, $7$, and $8$. To recover the correct topology, we should connect node $6$ to $5$ first. Then nodes $7$ and $8$ will get connected to $6$ naturally, as $d_{76} < d_{75}$ and $d_{86} < d_{85}$. However, consider $d_{65}$ and $d_{75}$,
\begin{align*}
    d_{65} = var(\mathbf{v}_6^a-\mathbf{v}_5^a) + var(\mathbf{v}_6^c-\mathbf{v}_5^c),\quad
    d_{75} = var(\mathbf{v}_7^c-\mathbf{v}_5^c)
\end{align*}
We have no guarantee that $d_{65} < d_{75}$ due to the presence of additional phase variance in $d_{65}$ illustrating how an algorithm that doesn't order nodes by decreasing number of phases may return an incorrect topology, unlike GPT.

\subsection{Alternative Estimation Scenarios}
In the general scenario, GPT recovers both phase and topology from voltage measurements. Our theoretical results also establish estimation methods for restricted settings:
\subsubsection{Phase Identification with Topology Information}
If topology is known, phases can be identified by greedily matching adjacent nodes using $c_{ij}^{\mathcal{O}}$ (\ref{eq:covarv}) across edges $ij \in \mathcal{E}$. \cite{olivier2018phase, pezeshki2012consumer, arghandeh2015topology} similarly use the Pearson correlation coefficient of voltages as the distance, which is related to the covariance but not theoretically justified. GPT is highly local, unlike \cite{blakely2019spectral, wang2016phase, olivier2018phase, short2012advanced, pezeshki2012consumer} which cluster all nodal voltages to recover phase. K-means is a popular clustering algorithm choice \cite{wang2016phase}. However, even if the correct phase matching is the globally optimal solution of the k-means cost, the optimization is non-convex and may not converge to the global minima. Our greedy approach, however, is guaranteed to result in the optimal solution.

\subsubsection{Topology Estimation with Phase Information}
If phase labels are known, $d_{ij}$ (\ref{eq:sovDist}) can be directly minimized to recover topology and GPT reduces to greedy spanning tree learning generalizing prior work for the single phase case \cite{deka2017structure}. Compared to \cite{liao2019unbalanced, deka2019topology} that use conditional independence tests and need matrix inversions, GPT has improved sample performance, as demonstrated in Section \ref{sec:Results}. 

\subsubsection{Estimation using voltage magnitudes only}
While GPT is based on nodal voltage phasors, it can also use voltage magnitudes $v_i^{\phi} = |\mathbf{v}_i^{\phi}|$. This is theoretically justified by linearizing (\ref{eq:line_impedance}) for line $kl$:
\begin{align*}
    \mathbf{I}_{kl} &= \mathbf{Y}_{kl}\begin{bmatrix}e^{j\theta^a}(v_k^ae^{j\underline{\theta}_k^a}-v_l^ae^{j\underline{\theta}_l^a}) \\ e^{j\theta^b}(v_k^be^{j\underline{\theta}_k^b}-v_l^be^{j\underline{\theta}_l^b}) \\ e^{j\theta^c}(v_k^ce^{j\underline{\theta}_k^c}-v_l^ce^{j\underline{\theta}_l^c})  \end{bmatrix}\nonumber\\
    &\approx \mathbf{Y}_{kl}\mathbf{D}_r((V_k-V_l)+j(\underline{\theta}_k-\underline{\theta}_l))
\end{align*}
where $\theta^{\phi}$ is the phase $\phi$ reference angle and $\underline{\theta}_k^{\phi} = \theta_k^{\phi} - \theta^{\phi}$. The linearization assumes small magnitude deviations from the reference and small angle difference between neighboring nodes. Properties of voltage magnitudes across the network can then be derived under assumptions on  $\mathbf{Y}_{kl}\mathbf{D}_r, \mathbf{I}_{kl}$ to obtain GPT for phase and topology recovery. 

%%%%%%%%%%%%%%%%%%%% Topology Recovery Error
\begin{figure*}[ht]
\centering     %%% not \center
\includegraphics[width=\textwidth]{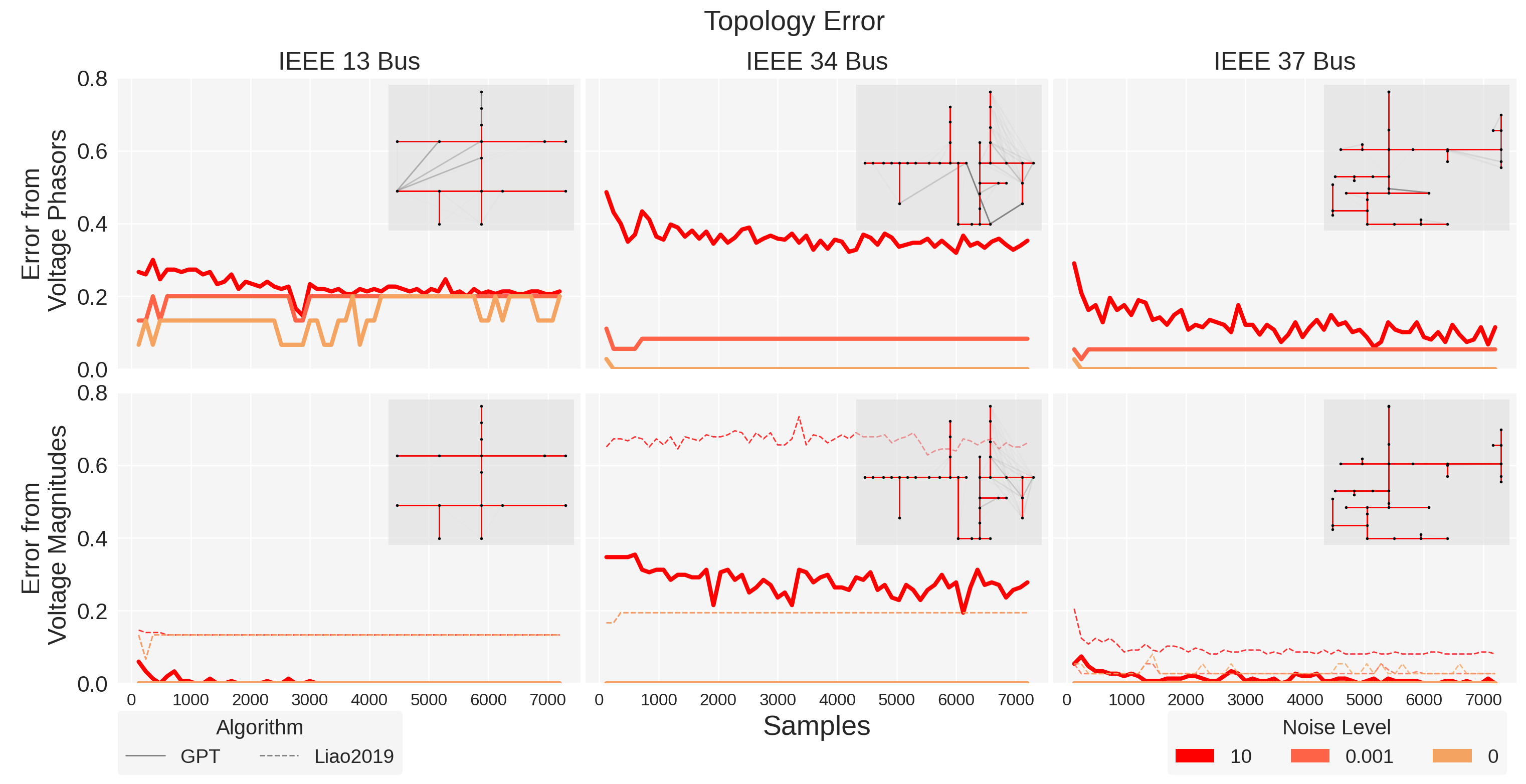}
\caption{GPT topology error vs number of samples for three test feeders and three noise levels. Samples are assumed to arrive at $120 Hz$. Insets show estimated (grey) and true (red) network lines across trials, with the opacity of grey lines indicating how many times the edge was recovered. GPT is evaluated on voltage phasor and magnitude data, with performance on magnitudes compared to state of the art in Liao2019 \cite{liao2019unbalanced}.}
\label{fig:top_results}
\end{figure*}
%%%%%%%%%%%%%%%%%%%%%%%%%%%%%%%%%%
\begin{figure}[ht]
\centering     %%% not \center
\includegraphics[width=\columnwidth]{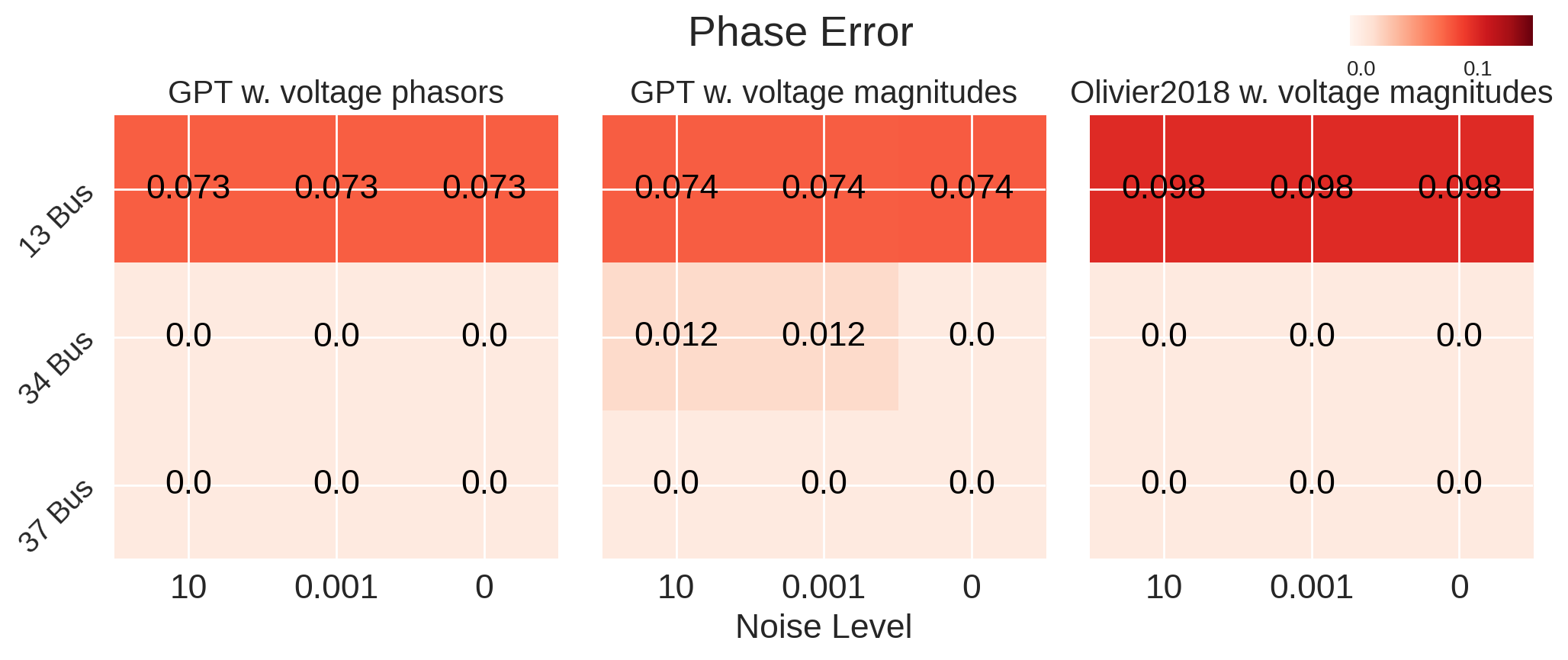}
\caption{GPT phase matching error for three test feeders and three noise levels using voltage phasor and magnitude data. The last table allows comparison with the state of the art phase matching method of Olivier2018 \cite{olivier2018phase}. Mostly, GPT has comparable performance to Olivier2018, but significantly outperforms it on the 13 bus system.} 
\label{fig:phase_results}
\end{figure}
%%%%%%%%%%%%%%%%%%%%%%%%%%%%%%%%%%
\begin{figure}[ht]
\centering     %%% not \center
\includegraphics[width=\columnwidth]{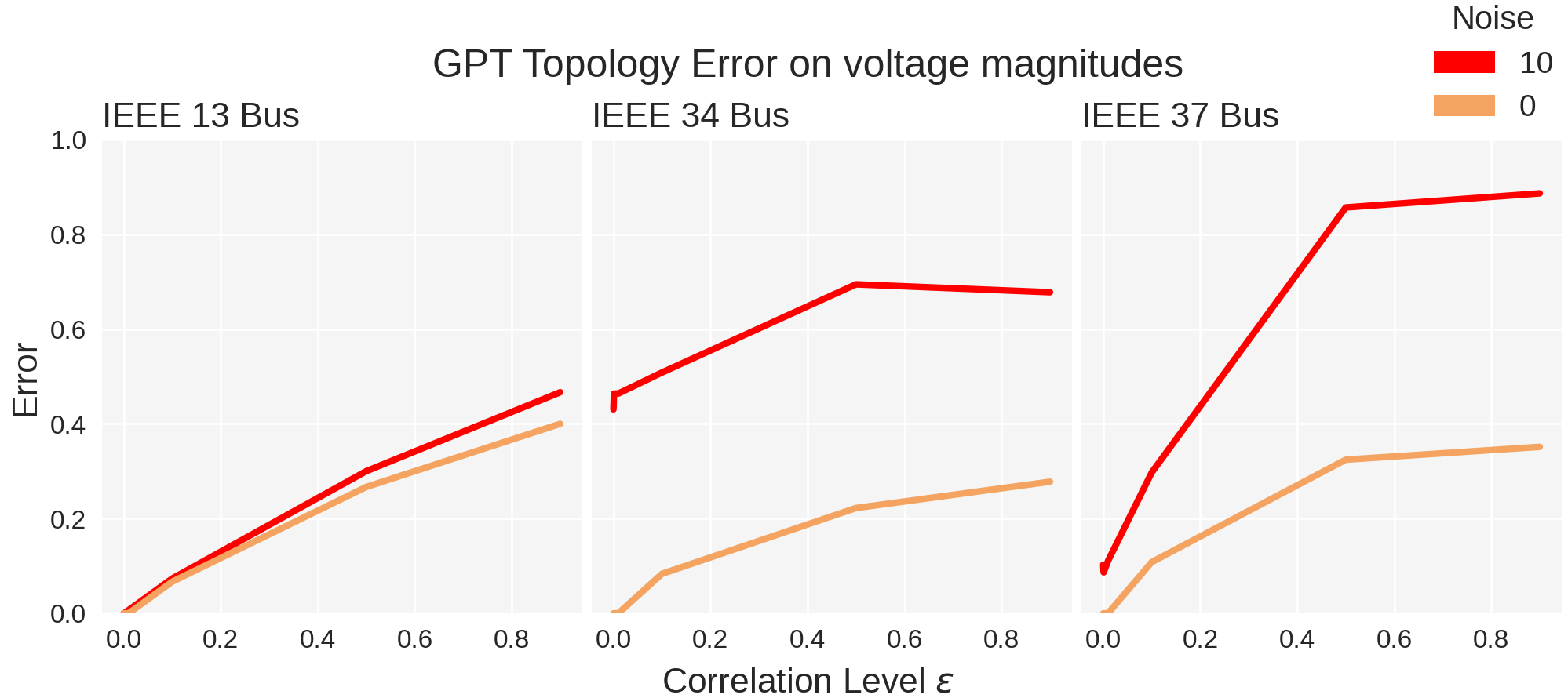}
\caption{GPT topology error from voltage magnitudes as injections become increasingly correlated ($\epsilon \rightarrow 1$) for three test feeders and two noise levels. Notice the rapid rise in error with increasing correlation particularly for the largest network.}
\label{fig:corr_results}
\end{figure}
%%%%%%%%%%%%%%%%%%%%%%%%%%%%%%%%%%
\section{Simulation Experiments}\label{sec:Results}
We present simulation results of GPT. We measure average errors in phase and topology recovery, normalized by network size: 
\begin{align*}
    &\textnormal{Topology Error} = \frac{\textnormal{wrong edges} + \textnormal{missing edges}}{\textnormal{total edges}},\\
    &\textnormal{Phase Error} = \frac{\textnormal{wrong nodal phases}}{\textnormal{total nodal phases}}
\end{align*}
Further, we evaluate the algorithm's sensitivity to the following parameters. 
\begin{itemize}[leftmargin=*]
    \item \emph{Measurement noise}: We add white noise $\mathbf{n}$ to original measurement $\mathbf{v}_i$: $\tilde{\mathbf{v}}_i= \mathbf{v}_i + \mathbf{n}$, defining $
        \textnormal{noise level}(\tilde{\mathbf{v}}_i) = \frac{var(\mathbf{n})}{var(\mathbf{v}_i)}
    $. As GPT uses voltage covariances, it depends on relative precision and not absolute accuracy and is immune to the stable transducer errors that afflict distribution PMU data \cite{brady2016uses}.
    \item \emph{Number of measurement samples}: Assuming $120$ Hz distribution PMU measurements \cite{von2014micro}, we record performance on $1$ second to $1$ minute of voltage data. 
    \item \emph{Load Correlations}: We test GPT's sensitivity to the assumption of uncorrelated injections, by varying the correlations of the loads while maintaining their variance. This is done by setting load covariance matrix $\Sigma = \sigma^2 ((1-\epsilon) \mathbb{I} + \epsilon11^T)$. As $\epsilon \rightarrow 1$, injections become more correlated. 
\end{itemize}

Three IEEE distribution test networks are simulated in OpenDSS: the 13 and 34 bus networks have some one and two phases buses, while the 37 bus network has all three phase buses \cite{kersting1991radial}. We modify the models by adding loads at every bus, and by disabling voltage regulators, which invalidate the assumption of voltages driven by injections. We fluctuate the load injections at each phase at each bus, and simulate the network with power flow to obtain non-linear voltages. 

Fig. \ref{fig:top_results} plots topology recovery accuracy for three noise levels ranging from $0$ (no noise) to $10$, with $1$ second to $1$ minute of voltage magnitude or phasor measurements. PMUs are highly precise; and the noise level would realistically be $\sim 0.001$ \cite{brown2016characterizing, bariya2019empirical}. Nevertheless, GPT performs well under more noise as measurement samples increase. For all test networks and measurement durations, GPT achieves perfect topology recovery from voltage magnitudes for $0$ and $0.001$ noise. Insets in Fig. \ref{fig:top_results} show recovered topologies across trials. Note how errors are localized to a few nodes, and lower for voltage magnitudes. Fig \ref{fig:top_results} also compares performance on voltage magnitudes to Liao2019 \cite{liao2019unbalanced}, showing that GPT outperforms it across scenarios. 

Fig. \ref{fig:phase_results} presents GPT's phase matching error on the same three networks and noise levels averaged across several sample durations (we found phase matching error to be invariant to sample duration). The error is compared to that of the approach in Olivier2018 \cite{olivier2018phase}. The methods have comparable performance, except on the 13 bus network, where GPT outperforms Olivier2018 across SNRs. 

Fig. \ref{fig:corr_results} shows topology recovery sensitivity as injections stray from the uncorrelated assumption (\ref{eq:assumption1}). Error increases rapidly as loads become more correlated. In reality, over short time durations, it is reasonable to assume that injections will be uncorrelated across nodes or can be de-trended \cite{liao2019unbalanced}. We use at most one minute of data to recover phase and topology: short enough that the uncorrelated assumption should hold well.

Our polynomial time algorithms are suitable for real time application, taking on the order of seconds to recover phase and topology for the IEEE test cases shared here. On the largest 37 bus test case, the algorithm completes in $15$ seconds. 

\section{Conclusion}
We presented GPT for joint phase and topology identification from voltage measurements in unbalanced three phase networks where each bus can have single, two, or three phases. GPT is polynomial time and relies on analytical trends in nodal voltage statistics. We proved its theoretical correctness under reasonable assumptions on load statistics and line impedances, demonstrating its efficacy on non-linear voltages from three test feeders simulated under realistic conditions. GPT is robust to measurement non-idealities, and outperforms the prior work in both phase and topology recovery. 

Extensions of the work include handling of voltage regulators at some nodes and improving performance under load correlations. The linearized unbalanced three phase model presented here can motivate many future analytics, including extensions to settings with limited measurements, where our derived distance metrics could give a sense of the relative proximity of dispersed measurement points if not the precise topology. It is challenging to directly extend this work to non-radial settings as inter-nodal distances become difficult to define precisely in meshed networks. However, the model could inform graph theoretic monitoring approaches of meshed networks, similar to those in \cite{bariya2020physically}.

\bibliographystyle{ieeetr}
\bibliography{bibliography}

\begin{IEEEbiography}
    [{\includegraphics[width=1in,height=1.25in,clip,keepaspectratio]{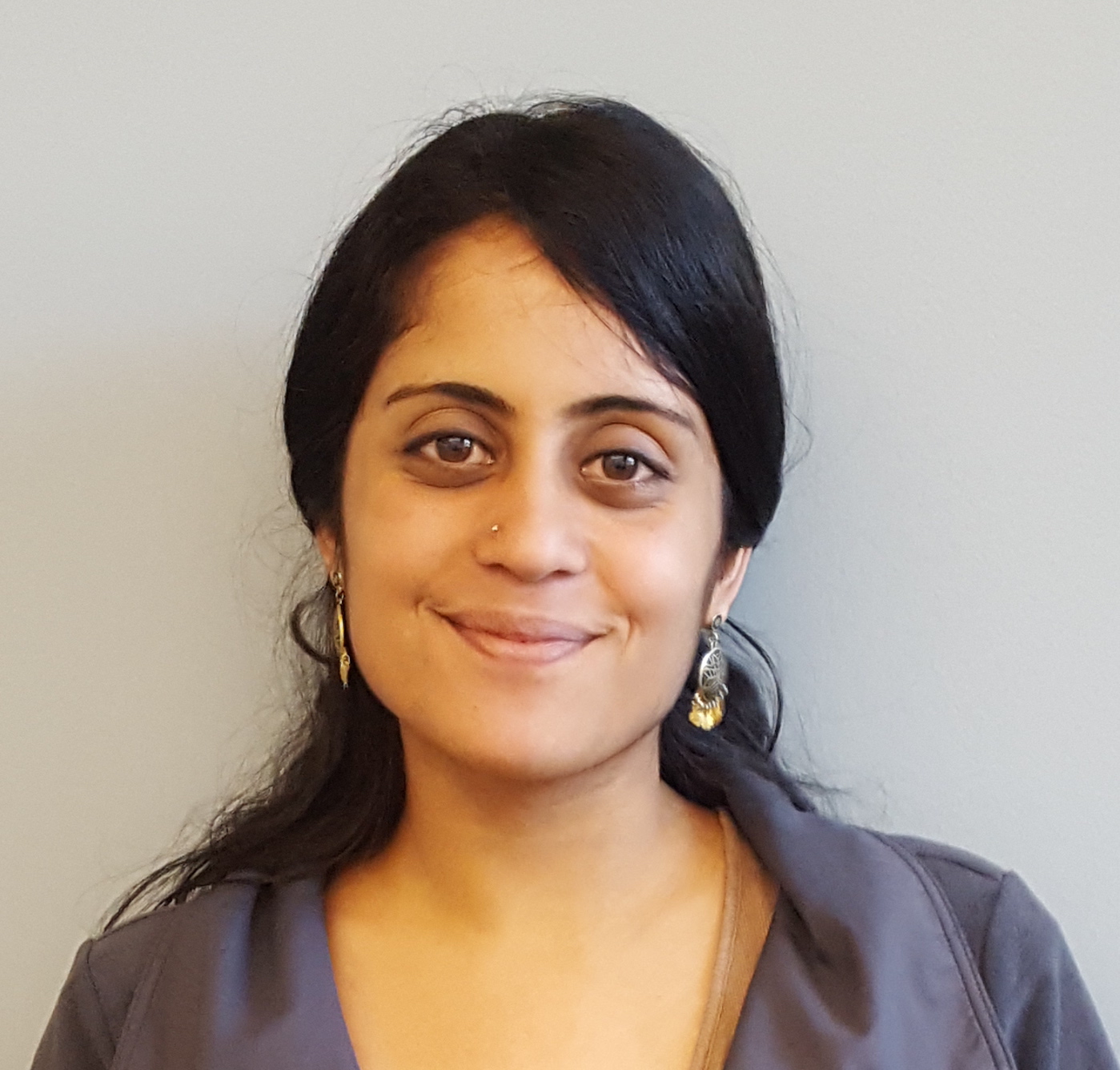}}]{Mohini Bariya} is a doctoral candidate in the Department of Electrical Engineering
and Computer Science at the University of California, Berkeley. Her work is on grid modernization and renewables integration, with a focus on using sensor measurements for system awareness to enable
safer, more efficient grid operations. She received her B.A. in computer science (2016) from UC Berkeley.
\end{IEEEbiography}
\vskip -2\baselineskip
\begin{IEEEbiography}
    [{\includegraphics[width=1in,height=1.25in,clip,keepaspectratio]{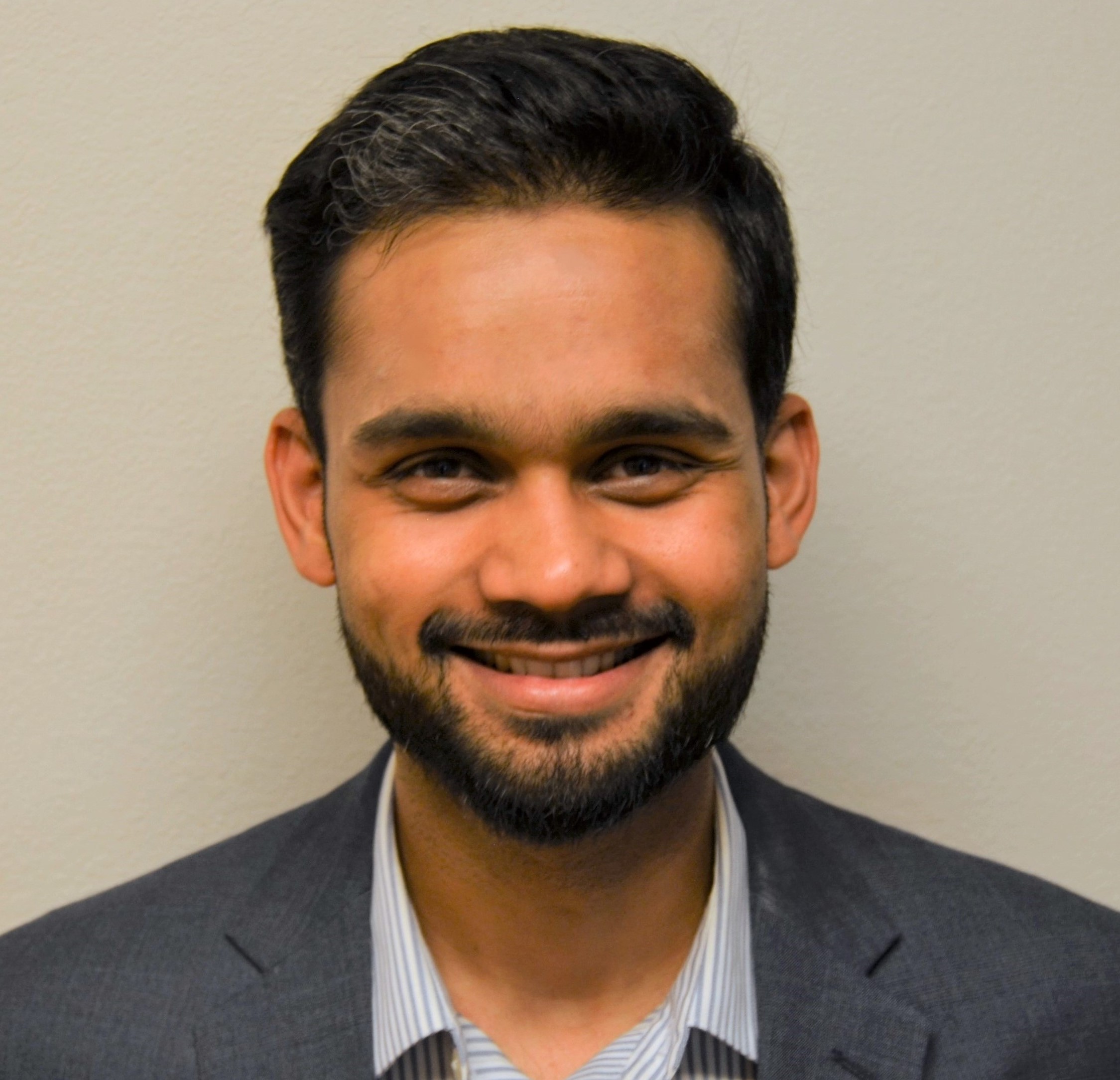}}]{Deepjyoti Deka (SM'20)} is a staff scientist in the Theoretical Division at Los Alamos National Laboratory, where he was previously a postdoctoral research associate at the Center for Nonlinear Studies. His research interests include data-analysis of power grid structure, operations and security, and optimization in social and physical networks. At LANL, Dr. Deka serves as a co-principal investigator for DOE projects on machine learning in power systems and in cyber-physical security. Before joining LANL, he received the M.S. and Ph.D. degrees in Electrical and Computer Engineering from the University of Texas, Austin, TX, USA, in 2011 and 2015, respectively. He completed his undergraduate degree in Electronics and Communication Engineering from IIT Guwahati, India with an institute silver medal as the best outgoing student of the department in 2009.
\end{IEEEbiography}
 \vskip -2\baselineskip
\begin{IEEEbiography}
    [{\includegraphics[width=1in,height=1.25in,clip,keepaspectratio]{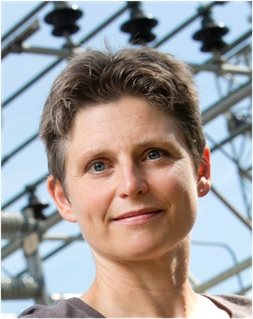}}]{Alexandra “Sascha” von Meier (M'10)} is an adjunct professor in the Department of Electrical Engineering and Computer Science at the University of California, Berkeley and directs the Electric Grid Research program at the California Institute for Energy and Environment (CIEE). She is also a faculty scientist at the Lawrence Berkeley National Laboratory. Her research is driven by the vision of a nimble, adaptable and resilient electric power infrastructure that optimally recruits intermittent renewable resources, energy storage and electric demand response. She holds a B.A. in Physics (1986) and a Ph.D. in Energy and Resources (1995) from UC Berkeley.
\end{IEEEbiography}

%%%%%%%%%%%%%%%%%%%%%%%
\end{document}